\documentclass[11pt]{article}
\usepackage{times}
\usepackage{fullpage}
\usepackage{epsf}
\usepackage{amsmath,amsfonts}
\newtheorem{proposition}{Proposition}
\newtheorem{lemma}[proposition]{Lemma}

\newtheorem{corollary}[proposition]{Corollary}
\newtheorem{theorem}[proposition]{Theorem}

\newcommand\req[1]{(\ref{#1})}

\newcommand{\qed}{\hbox{\rule{6pt}{6pt}}}
\newenvironment{proof}[1][]{\paragraph{Proof{#1}}}{\hfill\qed\medskip\\}

\newcommand\drop[1]{}

\def\showlabel#1{}
\def\showfiglabel#1{}

\def\junk#1{}
\def\journal#1{}
\newcommand\packlist{\setlength{\itemsep}{1pt}
\setlength{\parskip}{0pt}\setlength{\parsep}{0pt}}

\newcounter{invari}
\newenvironment{invariants}{\begin{enumerate}\packlist\setcounter{enumi}{\value{invari}}
\renewcommand\theenumi{(\roman{enumi})}
\renewcommand\labelenumi\theenumi}{\setcounter{invari}{\value{enumi}}\end{enumerate}}

\newcommand\cutcolor{\textnormal{\sf cut-or-color}}

\newcommand\wt{\widetilde}
\newcommand\tO{{\widetilde O}}
\newcommand\tOmega{{\widetilde\Omega}}
\newcommand\tTheta{{\widetilde\Theta}}
\newcommand\ds{\Delta}
\newcommand\dt{\delta}
\newcommand\nh{\Psi}
\newcommand\tcite[1]{~\cite{#1}}
\newcommand\tref[1]{~\ref{#1}}

\begin{document}
\setcounter{page}0

\title{Combinatorial coloring of 3-colorable graphs}

\author{{\em Ken-ichi Kawarabayashi}\thanks{Ken-ichi Kawarabayashi's
    research is partly supported by Japan Society for the Promotion of
    Science, Grant-in-Aid for Scientific Research, by C \& C
    Foundation, by Kayamori Foundation and by Inoue Research Award for
    Young Scientists.}\\National Institute of Informatics, Tokyo, Japan\\
  \texttt{k\_keniti@nii.ac.jp}
  \and {\em Mikkel Thorup}\\
  University of Copenhagen and AT\&T Labs---Research\\
  \texttt{mikkel2thorup@gmail.com}}
%\author{Ken-Ichi Kawarabayashi
%\thanks{Research partly
%supported by C \& C Foundation, by
%Kayamori Foundation and by Inoue Research Award.}
%\and Mikkel Thorup}

\date{\ \vspace{-1cm}}

\maketitle

\begin{abstract}
We consider the problem of coloring a 3-colorable graph
in polynomial time using as few colors as possible.
We present a combinatorial algorithm getting
down to $\tO(n^{4/11})$ colors. This
is the first combinatorial improvement of
Blum's $\tO(n^{3/8})$ bound from FOCS'90.  Like Blum's algorithm,
our new algorithm composes nicely with recent semi-definite
approaches. The current best bound is
$O(n^{0.2072})$ colors
by Chlamtac from FOCS'07. We now bring it
down to $O(n^{0.2038})$ colors.
\end{abstract}
%\vfill \baselineskip 11pt \noindent 16 December, 2008, revised \today.
%\category{G.2}{Discrete Math}{Combinatorics}
%A category including the fourth, optional field follows...
%\category{G.2.2}{Graph Theory}{Combinatorics}[Graph
%algorithms,Computations on discrete structures ]

%\terms{Algorithm, Theory}
%\medskip

\thispagestyle{empty}
\clearpage

\section{Introduction}
If ever you want to illustrate the difference between what we consider
hard and easy to someone not from computer science, use the example of
2-coloring versus 3-coloring: suppose there is too much fighting in a
class, and you want to split it so that no enemies end up in the same
group. First you try with a red and a blue group. Put someone in the
red group, and everyone he dislikes in the blue group, everyone they
dislike in the red group, and so forth. This is an easy systematic
approach. Digging a bit deeper, if something goes wrong, you have an
odd cycle, and it is easy to see that if you have a necklace with an
odd number of red and blue beads, then the colors cannot alternate
perfectly. Knowing that red and blue do not suffice, we might try
introducing green, but this is already beyond what we believe
computers can do.

Three-coloring is a classic NP-hard problem. It was
proved hard by Garey, Johnson, and Stockmeyer at
STOC'74 \cite{GJS76}, and was the prime example of NP-hardness
mentioned by Karp in 1975 \cite{Karp75}. It
is an obvious target for any approach to NP-hard
problems. With the approximation approach, given a 3-colorable graph,
we try to color it in polynomial time using as few colors as
possible. This challenge has engaged many researchers. At STOC'82,
Wigderson \cite{Wig83} got down to $O(n^{1/2})$ colors for a graph
with $n$ vertices. Berger and Rompel \cite{BR90} improved this
to $O((n/(\log n))^{1/2})$. Blum \cite{Blum94} came with the first
polynomial improvements, first to $\tO(n^{2/5})$ colors at STOC'89,
and then to $\wt O(n^{3/8})$ colors at FOCS'90.

The next big step at FOCS'94 was by Karger, Motwani, Sudan \cite{KMS98}
who used semi-definite programming (SDP). For
a graph with maximal degree $\Delta_{\max}$, they
got down to $O(\Delta_{\max}^{1/3})$ colors. Combining
it with Wigderson's algorithm, they got down to $O(n^{1/4})$ colors.
Later Blum and Karger \cite{BK97} combined the SDP \cite{KMS98}
with Blum's \cite{Blum94} algorithm, yielding
an improved bound of $\wt O(n^{3/14}) =\tO(n^{0.2142})$ (actually we
can get $n^{0.2142}+2$ colors since we have rounded the exponent up). Later
improvements in semi-definite programming have also
been combined with Blum's algorithm. At STOC'06, Arora, Chlamtac, and
Charikar \cite{ACC06} got down to $O(n^{0.2111})$ colors. The proof in \cite{ACC06} is based on
the seminal result of Arora, Rao and Vazirani \cite{ARV09} which gives an $O(\sqrt{\log n})$ algorithm for the sparsest cut problem. The last
improvement was at FOCS'07 by Chlamtac \cite{Chl07} who
got down to $O(n^{0.2072})$ colors.

Only a few lower bounds are known. The
strongest known hardness result shows that it is NP-hard to get
down to $5$ colors \cite{VGK04,KLS00}.
Recently, Dinur, Mossel and Regev \cite{DMR09}
showed that it's hard to color with any constant number of colors (i.e.,
$O(1)$ colors) based on a variant of Unique Games Conjecture.

Some integrality gap results \cite{FLS02,KMS98,Sze94} show that the
simple SDP relaxation has integrality gap at least $n^{0.157}$.  It
is therefore natural to go back and see if we can improve things
combinatorially.

In this paper, we present the first improvement on the combinatorial
side since Blum in 1990 \cite{Blum94}. With a purely
combinatorial approach, we get
down to $\wt O(n^{4/11})$ colors. Combining it with Chlamtac's
SDP \cite{Chl07}, we get down to $O(n^{0.2038})$ colors..

\paragraph{Techniques}
In the details we reuse a lot of the techniques pioneered by Blum
\cite{Blum94}, but our overall strategy is more structural. We will be
looking for sparse cuts that we can recurse over. When no more sparse
cuts can be found (and if we are not done by other means), we will
have crystallized a vertex set $X$ that is guaranteed to be
monochromatic in some 3-coloring.  The vertices in $X$ can then be
identified. We note for comparison that one of the main technical
lemmas in Blum \cite{Blum94} (see Lemma \ref{lem:blum} below) is a test to
see if a sufficiently large set $Y$ is multichromatic in the
sense that no color can dominate too much in any coloring. The
monochromatic set $X$ that we identify is typically much too small for
Blum's lemma to apply.

Below we focus on our combinatorial algorithm. The integration
with SDP is essentially explained in \cite{BK97} and is sketched
in the end.

\section{Preliminaries including ingredients from Blum}\label{sec:blum}
We are given a 3-colorable graph $G=(V,E)$ with $n=|V|$ vertices.
For a vertex set $X \subseteq V$, let $N(X)$ be the neighborhood of $X$
and  $G|X$ be the subgraph induced by $X$.
The unknown 3-coloring is with red, green, and blue.
If we say a vertex set \emph{$X$ is green}, we mean that every
vertex in $X$ is colored by green.
\drop{Given two vertex sets $S, T \subset V$ with $S \cap T =\emptyset$, the \emph{degrees} in $T$ (or from $T$ to $S$) mean
degrees of all the vertices in $T$ to $S$.   }

For some color target $k$
which is polynomial in $n$, we wish to find
a $\tO(k)$ coloring of $G$ in polynomial time.  We are going to reuse several ideas and techniques from
Blum's approach\tcite{Blum94}.
\paragraph{Progress}
Blum has a general notion of {\em progress towards $\tO(k)$ coloring}
(or {\em progress} for short if $k$ is understood) which is defined
such that if we for any 3-colorable graph can make progress towards
$\tO(k)$ coloring in polynomial time, then we can $\tO(k)$ color
any 3-colorable graph in polynomial time.

Blum defines several types of progress, but here the only concrete
type of progress we need to know is that of {\em monochromatic progress\/}
where we identify a set
of vertices that is monochromatic in some 3-coloring, hence which can
be identified in a single vertex. Otherwise we will only make progress
via results of Blum presented below using a common parameter
\begin{equation}\label{eq:nh}
\nh = n/k^2.
\end{equation}
A very useful tool we get from Blum is the
following multichromatic test:
\begin{lemma}[{\cite[Corollary 4]{Blum94}}] \label{lem:blum} Given a vertex set $X\subseteq V$ of
size at least $\nh = n/k^2$, in polynomial time, we can either make
progress towards $\tO(k)$-coloring of $G$, or else guarantee that
under {\em every} legal 3-coloring of $G$, the set $X$ is
multichromatic.
\end{lemma}
In fact Blum has a stronger lemma \cite[Lemma 12]{Blum94}
guaranteeing not only that $X$ is multichromatic, but
that no single color is used by more than a fraction
$(1-1/(4\log n))$ of the vertices in $X$. This stronger 
version is not needed here. Using Lemma \ref{lem:blum} he proves:
\begin{lemma}[{\cite[Theorem 3]{Blum94}}]\label{lem:small-nh} If two vertices
have more then $\nh$ common neighbors, we can make
progress towards $O(k)$ coloring. Hence we can assume that
no two vertices have more than $\nh$ common neighbors.
\end{lemma}
Using this bound on joint neighborhoods, Blum proves the following
lemma (which he never states in this general quotable form):
\begin{lemma}\label{lem:large-neighborhoods}
If the vertices in a set $Z$ on the average have $d$ neighbors in $U$, then
the whole set $Z$ has at least $\min\{d/\nh,|Z|\}\,d/2$
distinct neighbors in $U$.
\end{lemma}
\begin{proof}
If $d/\nh\leq 2$, the result is trivial, so $d/\nh\geq 2$.
It suffices to prove the lemma for $|Z|\leq d/\nh$, for
if $Z$ is larger, we restrict our attention to the $d/\nh$ vertices
with most neighbors in $U$. Let the vertices in $Z$ be ordered
by decreasing degree into $Z$. Let $d_i$ be the degree of
vertex $v_i$ into $U$. We now study how the neighborhood of $Z$ in $U$  
grows as we include the vertices $v_i$. When we add $v_i$,
we know from Lemma\tref{lem:small-nh} that its joint neighborhood with any
(previous) vertex $v_h$, $h<i$ is of size at most $\nh$. It follows that
$v_i$ adds at least $d_i-(i-1)\nh$ new neighbors in $U$, so
$|N(Z)\cap U|\geq \sum_{i=0}^{|Z|-1} (d_i-(i-1)\nh)>|Z|d/2$.
\end{proof}

\paragraph{Second neighborhood structure}
Let $\Delta_{\min}$ be the smallest degree in graph $G$. With color
target $k$, we can trivially assume $\Delta_{\min}> k$ (but
$\Delta_{\min}$ may be much larger if we apply SDP to low degree
vertices). For some $\Delta_0=\tOmega(\Delta_{\min})$, Blum
\cite[Theorems 7 and 8 and the Proof of Theorem 5]{Blum94}
identifies in polynomial time a subgraph $H_0$ of $G$ that consists of
a vertex $r_0$ and two disjoint vertex sets $S_0,T_0$, with the following
3-level structure:
\begin{itemize}\packlist
\item A root vertex $r_0$. We assume $r_0$ is colored red in any
3-coloring.
\item A first neighborhood $S_0\subseteq N_G(r_0)$ of size at least $\ds_0$.
\item A second neighborhood $T_0\subseteq N_G(S_0)$ of size at most $n/k$.
\item All edges in $H_0$ go between $r_0$ and $S_0$ and between $S_0$ and $T_0$.
\item The vertices in $S_0$ have average degree $\ds_0$ into $T_0$.
\item The degrees from $T_0$ to $S_0$ are all within a factor
$(1\pm o(1))$ around an average
$\dt_0\geq \ds_0^2\, k/n$.
\end{itemize}
Note that \cite[Theorems 7 and 8]{Blum94} does not have this size
bound on $T_0$. Instead there is a large set $R$ of red vertices
leading to a large identifiable independent set constituting a
constant fraction of $T_0$. If this set is of size $\tOmega(n/k)$,
then Blum makes progress towards a $\tO(k)$ coloring as described
in \cite[Proof of Theorem 5]{Blum94}, and we are done.
Assuming that this did not happen, we have the size bound on $T_0$.

Blum seeks progress directly in the above structure, but we are going
to apply a series of reductions which either make progress, find
a good cut recursing on one side, or identify a monochromatic set. This is
why we already now used the subscript $_0$ to indicate the original
structure provided by Blum \cite{Blum94}.

\section{Our coloring algorithm}
We will use Blum's second neighborhood structure $H_0$
with a color target
\begin{equation}\label{eq:def-k}
k=\tTheta\left((n/\ds_{\min})^{4/7}\right).
\end{equation}
We are going to work on induced subproblems
$(S,T)\subseteq (S_0,T_0)$ defined in terms of a subsets $S\subseteq
S_0$ and $T\subseteq T_0$. The edges considered in the subproblem
are exactly those from $H_0$ between $S$ and $T$. This
edge set is denoted $E_0(S,T)$.
With $r_0$ red in any 3-coloring, we know that all
vertices in $S\subseteq S_0$ are blue or green.
We will define {\em high degrees\/} in $T$ (to $S$) as degrees bigger than
$\dt_0/16$ (almost a factor 16 below the average in $T_0$), and we will make sure that any subproblem $(S,T)$ considered
satisfies:
\begin{invariants}
\item\label{inv:high-degrees} We have more than $\nh$ vertices of high degree
in $T$.
\end{invariants}
\paragraph{Cut-or-color}
We will implement a subroutine $\cutcolor(t,S,T)$ which for a problem
$(S,T)\subseteq (S_0,T_0)$ takes starting point in an arbitrary high degree vertex $t\in
T$. It will have one of the following outcomes:
\begin{itemize}\packlist
\item Reporting a ``sparse cut around a subproblem $(S',T')\subseteq(S,T)$''
with no cut edges
between $S'$ and $T\setminus T'$ and only few cut edges
between $T'$ and $S\setminus S'$. The exact definition of
a sparse cut is complicated, but at this point, all we need to
know is that $\cutcolor$ may declare a sparse cut.
\item Some progress toward $k$-coloring with Blum's Lemma \ref{lem:blum}.
\item A guarantee that if $r$ and $t$ have the different colors in
some 3-coloring $C_t$ of $G$, then $S$ is monochromatic in $C_t$.
\end{itemize}
We note that testing whether or not $S$ can be monochromatic is only
new if $|S|<\nh$, for if $|S|\geq\nh$ and $S$ was
monochromatic, we would get immediate progress with Lemma
\ref{lem:blum}.

\paragraph{Recursing towards a monochromatic set}
Using $\cutcolor$, we now describe our main recursive algorithm which
takes as input a subproblem $(S,T)$.  Let $U$ be the set of high
degree vertices in $T$. By \ref{inv:high-degrees} we have
$|U|\geq\nh$, so we can apply Blum's multicolor test from Lemma
\ref{lem:blum}.  Assuming we did not make progress, we know that $U$ is
multichromatic in every valid 3-coloring.  We now apply $\cutcolor$ to
each $t\in U$, stopping only if a sparse cut is found or progress is
made.  If we make progress, we are done, so assume that this does not
happen.
\paragraph{Monochromatic case}
The most interesting case is if we get neither progress nor a sparse cut.
\begin{lemma}\label{lem:mono} If $\cutcolor$ does not find progress
or a sparse cut for any high degree $t\in U$, then
$S$ is monochromatic in some 3-coloring of $G$.
\end{lemma}
\begin{proof}
Since $U$ is multichromatic in any $3$-coloring, there
is a $3$-coloring $C_t$ where some $t\in U$ has
a different color than $r_0$, and then $\cutcolor$ guarantees
that $S$ is monochromatic in $C_t$.
\end{proof}
Thus we have found a monochromatic set, so monochromatic progress can
be made.
\paragraph{Sparse cut}
If a sparse cut around a subproblem $(S',T')$ is reported, we
recurse on $(S',T')$.

\section{Implementing $\cutcolor$}
We are now going to implement $\cutcolor$. The first part of it
is essentially the coloring that Blum \cite[\S 5.2]{Blum94}
uses for dense regions. We shall describe how we bypass the
limits of his approach as soon as we have presented his
part.

Assume that $t \in U$ and
$r$ have different colors in some coloring $C_t$; otherwise the
algorithm provides no monochromatic guarantee. Let us say that $r_0$ is red and $t$
is green.

Let $X$ be the neighborhood of $t$ in $S$ and let $Y$ be the
neighborhood of $X$ in $T$. As in \cite{Blum94} we note that all of
$X$ must be blue, and that no vertex in $Y$ can be blue. We are going
to expand $X\subseteq S$ and $Y\subseteq T$ preserving the following
invariant:
\begin{invariants}
\item\label{inv:XY} If $t$ was green and $r_0$ was red, then $X$
  is all blue and $Y$ has no blue.
\end{invariants}
%Note that when $X$ is blue, we can safely include $N(X)$ in
%$Y$.
If we end up with $X=S$, then \ref{inv:XY} implies that $S$ is
monochromatic in any 3-coloring where $r_0$ and $t$ have different colors.

\paragraph{$X$-extension}
Now consider any vertex $s\in S$ whose degree into $Y$ is at least $\nh$.
Using Lemma \ref{lem:blum} we can check that $N(s)\cap Y$ is
multichromatic. Since $Y$ has no blue, we conclude that
$s$ has both red and green neighbors, hence that $s$ is blue.
Note conversely that if $s$ was green, then all its neighbors in $Y$
would have to be red, and then the multichromatic test from
Lemma \ref{lem:blum} would have made progress.
Preserving \ref{inv:XY}, we now add the blue $s$ to $X$ and all neighbors of $s$ in $T$ to
$Y$.  We shall refer to this as an $X$-extension.

\paragraph{Relation to Blum's algorithm} Before continuing, let
us briefly relate to Blum's \cite{Blum94} algorithm.  The above
$X$-extension is essentially the coloring Blum \cite[\S 5.2]{Blum94}
uses for dense regions. He applies it directly to his structure $H_0$
from Section \ref{sec:blum}. He needs a larger degree $\dt_0\geq\ds_0^2k/n$
than we do, but then he proves that the set of vertices $s$ with degree
at least $\nh$ into $Y$ is more than $\nh$. This means that either he
finds progress with a green vertex $s$, or he ends up with a blue set
$X$ of size $\nh$, and gets progress applying Lemma \ref{lem:blum} to
$X$.

Our algorithm works for a smaller $\dt_0$ and thereby
for a smaller color target $k$. Our extended $X$ is typically
too small for Lemma \ref{lem:blum}. In fact, as we recurse,
we will get sets $S$ that themselves are much smaller than $\nh$. Otherwise
we would be done with Lemma \ref{lem:blum} if $S$ was monochromatic.

Below we introduce $Y$-extensions. They are similar in spirit to
$X$-extensions, and would not help us if we like Blum worked
directly with $H_0$. The important point will be that if we do not end
up with $X=S$, and if neither extension is possible, then we have identified
a sparse cut that we can use for recursion. We are thus borrowing from
Blum's proof \cite{Blum94} in the technical details, but the overall strategy, seeking sparse
cuts for recursion to crystallize a small monochromatic set $S$, is
entirely different (and new).

\paragraph{$Y$-extension}
We now describe a $Y$-extension, which is similar in spirit to the $X$-extension, but which
will cause more trouble in the analysis.  Consider a vertex $t'$ from
$T\setminus Y$. Let $X'$ be its neighborhood in $S$ (note that $X' \cap X=\emptyset$) and $Y'$ be the
neighborhood of $X'$ in $T$. Suppose $|Y\cap Y'|\geq\nh$. Using Lemma
\ref{lem:blum} we check that $Y\cap Y'$ is multichromatic. We claim
that $t'$ cannot be blue, for suppose it was. Then its neighborhood
would have no blue and $S$ is only blue and green, so $X'$
is all green.  Then $Y'$ has no green, but
$Y$ has no blue, so $Y'\cap Y$ would be all red, contradicting that
$Y\cap Y'$ is multichromatic. We conclude that $t'$ is not blue. Preserving \ref{inv:XY}, we now add $t'$ to $Y$.

\paragraph{Closure}
We are going to extend $X$ and $Y$ as long as possible. Suppose
we end up with $X=S$. With \ref{inv:XY} $\cutcolor$ declares
that $S$ is monochromatic in any
3-coloring where $r$ and $t$ have different
colors.

Otherwise we are in a situation where no $X$-extension nor $Y$-extension
is possible, and then
$\cutcolor$ will declare a sparse cut around
$(X,Y)$. A {\em sparse cut around $(X,Y)$\/} is simply
defined as being obtained this way. It has the following properties:
\begin{invariants}
\item\label{inv:high-start} The the original high degree
vertex $t$ has all its neighbors from $S$ in $X$, that is,
$N(t)\cap S\subseteq X$.
\item\label{inv:XtoT-Y} There are no edges between $X$ and $T\setminus Y$. To
see this, recall that when an $X$-extension adds $s'$ to $X$, it includes all
its neighbors in $Y$. The $Y$-extension does not change $X$.
\item\label{inv:X-extension} Each vertex $s'\in S\setminus X$ has $|N(s')\cap Y|< \nh$.
\item\label{inv:Y-extension} Each vertex  $t'\in T\setminus Y$ has $|N(N(t'))\cap Y|< \nh$.
\end{invariants}
A most important point here is that this characterization of a sparse cut
does not depend on the assumption that $t$ and $r$ have different
colors in some 3-coloring. It only assumes that $X$ and $Y$
cannot be extended further.

\section{Correctness} It should be noted that the correctness
of $\cutcolor$ follows from \ref{inv:XY} wich is immediate from
the construction. The technical difficulty that remains is to ensure
that we never end up considering a subproblem with too few
high-degree vertices for \ref{inv:high-degrees}, hence where
we cannot apply Lemma \ref{lem:blum} to ensure that the
high degree vertices do not all have the same color as $r_0$.

\section{Degree constraints}
Before we can start our recursive algorithm, we need
some slightly different degree constraints from those provided
by Blum \cite{Blum94} described in Section \ref{sec:blum}:
\begin{itemize}\packlist
\item The vertices in $S_0$ have average degree $\ds_0$ into $T_0$.
\item The degrees from $T_0$ to $S_0$ are all within a factor
$(1\pm o(1))$ around the average
$\dt_0\geq \ds_0^2\, k/n$.
\end{itemize}
We need some initial degree lower bounds, which are obtained simply by
removing low degree vertices creating our first induced subproblem $(S_1,T_1)\subseteq (S_0,T_0)$. Starting from $(S_1,T_1)=(S_0,T_0)$,
we repeatedly remove vertices from
$S_1$ with degree to $T_1$ below $\ds_0/4$ and vertices from $T_1$ with
degree to $S_1$ below $\dt_0/4$ until there are no such low-degree vertices left. The process eliminates less than
$|S_0|\ds_0/4+|T_0|\dt_0/4=|E_0(S_0,T_0)|/2$ edges, so half the edges
of $E_0(S_0,T_0)$ remain in $E_0(S_1,T_1)$. We also note that the
average degree in $T$ remains above $\dt_0/2=2\dt_1$. The point is
that the average on the $T$-side only goes down when we remove low
degree vertices from $S_0$, and that can take away at most $1/4$ of
the edges.
With $\ds_1=\ds_0/4$ and $\dt_1=\dt_0/4$, we
get
\begin{itemize}\packlist
\item The degrees from $S_1$ to $T_1$ are at least $\ds_1$.
\item The degrees from $T_1$ to $S_1$ are between $\dt_1$ and $(1+o(1))\dt_0<5\dt_1$,
with an average above $2\dt_1$.
Note that $\dt_1=\dt_0/4\geq \ds_0^2\, k/(4n)=4\ds_1^2 k/n$.
\end{itemize}
A {\em high degree\/} in $T$ can now be restated as a degree to
$S$ above $\dt_1/4=\dt_0/16=\ds_1^2 k/n$.
\begin{lemma}\label{lem:dt-nh}
$\ds_1=\nh n^{\Omega(1)}$.
\end{lemma}
\begin{proof}
Recall that
$\ds_1=\ds_0/4=\tOmega(\Delta_{\min})$, so
from \req{eq:def-k} we get
$k=\tO\left((n/\ds_1)^{4/7}\right)$,
or equivalently, $\ds_1=\tOmega\left(n/k^{7/4}\right)$. Since
$\nh=n/k^2$, we conclude that
$\ds_1=\tOmega\left(\nh k^{1/4}\right)=\nh n^{\Omega(1)}$.
\end{proof}

\paragraph{Recursion}  Our recursion will start
from $(S,T)=(S_1,T_1)\subseteq (S_0,T_0)$. We will
ensure that each of subproblems $(S,T)$ considered satisfies
the following degree invariants:
\begin{invariants}
\item\label{inv:ds} Each vertex $v\in S$ has all its neighbors from $T_1$ in $T$, so the the degrees from $S$ to $T$ remain at least $\ds_1$.
This invariant follows immediately from sparse cut condition
\ref{inv:XtoT-Y}.
\item\label{inv:dt} The average degree in $T$ to $S$ is at least $\dt_1/2$.
\end{invariants}
The
following key lemma shows that the degree constraints imply that
we have enough high degree vertices in $T$ that we can
test if they are multichromatic with Lemma \ref{lem:blum}.
\begin{lemma}\label{lem:high-degrees} \ref{inv:ds} and \ref{inv:dt} imply
\ref{inv:high-degrees}, i.e., that we have more than $\nh$ vertices of
high degree in $T$.
\end{lemma}
\begin{proof}
If $h$ is the fraction of high degree vertices in $T$,
the average degree in $T$ is at most $h\,5\dt_1+(1-h)\dt_1/4$ which by
\ref{inv:dt} is at least $\dt_1/2$. Hence
$h=\Omega(1)$. By \ref{inv:ds} we have $|T|\geq\ds_1$,
so we have $\Omega(\ds_1)$ high degree vertices
in $T$. By Lemma \ref{lem:dt-nh} this is much more
than $\nh$ high degree vertices.
\end{proof}

\section{Maintaining degrees recursively}
All that remains is to prove that invariant \ref{inv:dt}
is preserved, i.e., that the average degree from $T$ to $S$ does
not drop below half the original minimum degree  $\delta_1$ from $T_1$ to
$S_1$. Inductively, when a sparse cut is declared around a new subproblem 
$(S',T')=(X,Y)\subseteq(S,T)$,
we can assume that \ref{inv:ds} and \ref{inv:dt} are satisfied for $(S,T)$ and
that  \ref{inv:ds} is satisfied for $(X,Y)$. It remains
to prove \ref{inv:dt}  for $(X,Y)$.

Below we first show that when a sparse cut is declared
around $(X,Y)\subseteq (S,T)$, then
$X$ cannot be too small. We later complement this by showing
that the total number of edges cut in the recursion cannot be too large.
\begin{lemma}\label{lem:large-T} $|Y|\geq \ds_1^2k^2/(8n)$.
\end{lemma}
\begin{proof}
If $t$ is the high degree vertex we started with in $T$, then
by\tref{inv:high-start}, we have the whole neighborhood $Z$
of $t$ in $S$ preserved in $X$. By definition of
high degree, $|Z|\geq \dt_1/4=\ds_1^2 k/n$. By
\ref{inv:ds} the degrees from $Z$ to $Y$ are at least $d=\ds_1/2$, so
$d/\nh=\ds_1/2\cdot k^2/n=o(|Z|)$.
Hence by Lemma\tref{lem:large-neighborhoods}, we have
$|N(Z)\cap Y|\geq d/\nh\cdot d/2=\ds_1^2/(8\nh)=\ds_1^2k^2/(8n)$.
\end{proof}
In our original problem $(S_1,T_1)$, each vertex $v\in Y$ we had
$\dt_1$ edges to $S_1$, so the original number of edges from $Y$ to $S_1$
was at least
\begin{equation}\label{eq:original-edges}
\dt_1\ds_1^2k^2/(8n).
\end{equation}
To prove \ref{inv:dt}, we argue that
at least half of these edges are between $Y$ and $X$.
This follows if we can prove that the total number of edges
cut is only half the number in \req{eq:original-edges}.

The following main technical lemma relates the number of new cut edges around
the subproblem $(X,Y)$ to the reduction $|T\setminus Y|$ in
the size of the $T$-side:
\begin{lemma}\label{lem:cut-loss}
The number of cut edges from $Y$ to $S\setminus X$ is bounded by
\begin{equation}\label{eq:loss}
|T\setminus Y|\,\frac{40\dt_1n^2}{\ds_1^2
  k^4}.
\end{equation}
\end{lemma}
\begin{proof}
First we note that from
\ref{inv:X-extension} we get a trivial bound of $\nh|S\setminus
X|$ on the number of new cut edges, but is not strong enough for
\req{eq:loss}. Here we use \ref{inv:Y-extension} we get
\begin{equation}\label{eq:double-sum}
\sum_{y\in Y} |N(N(y))\setminus Y|=
\sum_{t'\in T\setminus Y} |N(N(t'))\cap Y|\leq |T\setminus Y|\nh=
|T\setminus Y| n/k^2.
\end{equation}
We will now, for any $y\in Y$, relate
$|N(N(y))\setminus Y|$ to the number $|N(y)\setminus X|$ of edges cut
from $y$ to $S\setminus X$. Let $Z=N(y)\setminus X$.  By \ref{inv:XtoT-Y} we have that
$N(N(y))\setminus Y=N(Z)\setminus Y$. Consider
any vertex $v\in Z$. By \ref{inv:ds}
the degree from $v$ to $T$  is at least $\ds_1$. Since $v\not\in X$, by \ref{inv:X-extension}, the
degree from $v$ to $Y$ is at most $\nh$, and by Lemma \ref{lem:dt-nh},
$\nh=o(\ds_1)$. The degree from $v$ to $T\setminus Y$ is therefore at
least
$(1-o(1))\ds_1\geq\ds_1/2$. This holds for every $v\in Z$.
It follows by Lemma \ref{lem:large-neighborhoods}
that $|N(Z)\setminus Y|\geq \min\{(\ds_1/2)/\nh,|Z|\}\,\ds_1/4$. Relative
to $|Z|$, this is
\[\frac{|N(Z)\setminus Y|}{|Z|}\geq \min\left\{\frac{\ds_1/(2\nh)}{|Z|},
1\right\}\,\ds_1/4\]
From our original configuration $(S_1,T_1)$, we
know that all degrees in $T$ are bounded by $5\dt_1$
and this bounds the size of $Z=N(y)$.  Therefore
\[\frac{\ds_1/(2\nh)}{|Z|}\geq \frac{\ds_1 k^2/(2n)}{5\dt_1}=
\frac{\ds_1 k^2}{10\dt_1 n}.\]
Since $\dt_1\geq 4\ds_1^2k/n$, we have
\[\frac{\ds_1 k^2}{10\dt_1 n}
\leq \frac{\ds_1 k^2}{10(4\ds_1^2k/n)n}=
\frac{k}{40\ds_1}<1.\]
Therefore
\[\frac{|N(Z)\setminus Y|}{|Z|}\geq \min\left\{\frac{\ds_1/(2\nh)}{|Z|},
1\right\}\,\ds_1/4\geq \frac{\ds_1 k^2}{10\dt_1n}\,\ds_1/4
= \frac{\ds_1^2 k^2}{40\dt_1n}.\]
Recalling $Z=N(y)\setminus X$ and $N(Z)\setminus Y=N(N(y))\setminus Y$,
we rewrite the inequality as
\[|N(y)\setminus X|\leq \frac{40\dt_1n}{\ds_1^2 k^2}\, |N(N(y))\setminus Y|.\]
Using \req{eq:double-sum}, we now get the desired bound on
the number of cut edges from $Y$ to $S\setminus X$:
\[\sum_{y\in Y} |N(y)\setminus X|\leq
\frac{40\dt_1n}{\ds_1^2 k^2}\, \sum_{y\in Y} |N(N(y))\setminus
Y|=\frac{40\dt_1n}{\ds_1^2 k^2}\,|T\setminus Y|\, n/k^2=|T\setminus Y|\,\frac{40\dt_1n^2}{\ds_1^2
  k^4}.\]
\end{proof}
From Lemma\tref{lem:cut-loss} it immediately follows that the
total number of edges cut in the whole recursion is
at most
\begin{equation}\label{eq:total-loss}
|T_1|\,\frac{40\dt_1n^2}{\ds_1^2
  k^4}\leq \frac{40\dt_1n^3}{\ds_1^2
  k^5}.
\end{equation}
This should be at most half the
original number of edges from \req{eq:original-edges}. Thus we
maintain \ref{inv:ds} with an average degree of $\dt_1/2$ from $Y$ as
long as
\[\frac{40\dt_1n^3}{\ds_1^2
  k^5}\leq \frac{\dt_1\ds_1^2k^2}{16n}\]
or equivalently
\begin{equation}\label{eq:final}
k^7\geq 640\, (n/\ds_1)^4.
\end{equation}
Thus, if $k$ satisfies \req{eq:final}, then all our degree constraints
are maintained, which means that we will keep recursing over
sparse cuts until we either make progress towards a $\tO(k)$ coloring,
or end up with a provably monochromatic set $S$ on which we can make
monochromatic progress towards $\tO(k)$ coloring. Since
$\ds_1=\Omega(\ds_0)=\tOmega(\Delta_{\min})$, we can
pick $k$ as a function of $\Delta_{\min}$ and $n$
such that $k=\tO\left((n/\Delta_{\min})^{4/7}\right)$ and
such that \req{eq:final} will be satisfied. Thus we conclude
\begin{theorem}
\label{thm:main}
If a 3-colorable graph has minimum
degree $\Delta_{min}$, then we can make progress towards
$\tO\left((n/\Delta_{min})^{4/7}\right)$ coloring in polynomial time.
\end{theorem}
Since we make trivial progress for vertices of degree below $k$,
we can assume $\Delta_{min}\geq k$, and hence we can balance with $k=\tTheta(n^{4/11})$
for a purely combinatorial algorithm to obtain the following corollary.
\begin{corollary}
\label{cor:main}
A 3-colorable graph on $n$ vertices can be colored
with $\tO\left(n^{4/11}\right)$ colors in polynomial time.
\end{corollary}

\section{Integration with SDP}
We will now show how to combine our combinatorial algorithm
with the best SDP coloring of Chlamtac.

\begin{lemma}[{\cite[Corollary 16]{Chl07}}]\label{lem:chlamtac}
For any 3-colorable graph $G$ on $n$ vertices with
maximal degree $\Delta_{\max}\leq n^{0.6546}$, in polynomial
time, we can find an independent set of size
$\Omega(n/\Delta_{\max}^{0.3166})$. Hence
we can make progress towards $O(\Delta_{\max}^{0.3166})$ coloring.
\end{lemma}
The above statement fixes a small typo in the statement of
\cite[Corollary 16]{Chl07} which says
that the maximal degree should be below $\Delta_{\max}=n^{0.6546}$,
as if we wouldn't benefit from a smaller $\Delta_{\max}$.

We note that combination with Theorem \ref{thm:main} would be trivial if
the $\Delta_{\min}$ in Theorem \ref{thm:main} was equal to the $\Delta_{\max}$ in
Lemma \ref{lem:chlamtac}. Things are not that simple, but
Karger and Blum \cite{BK97} have already shown how to combine
the original SDP of Karger et al. \cite{KMS98}
and Blum's algorithm \cite{Blum94}.

To describe the combination in our case, we first need to elaborate
a bit on Blum's progress from Section\tref{sec:blum}. We already
mentioned monochromatic progress where we identify a set
of vertices that we know are monochromatic in some coloring.

Less immediate types of progress work as follows. We
start with a graph $G=G_0$. We have
a constant number of players $i=1,...,\ell=O(1)$ that can announce progress
towards $\tO(k)$ coloring. Each player $i$ starts with
an empty vertex set $V_i$. When player $i$ announces
progress, he removes some vertices from $G$ and place
them in his set $V_i$. His promise is that if he ends up with
$n_i=|V_i|\geq n/(2\ell)=\Omega(n)$ vertices, then he
can find an independent set $I_i$ of size $n_i/\tO(k)$ in
the subgraph $G_0|V_i$ of $G_0$ induced by $V_i$.

The players play until $G$ has lost half its vertices. This means that
we only play on $G$ when $G$ has at least $n/2$ vertices left. When
they play, we always need someone to claim progress. Combined the players
end up taking more than $n/2$ vertices, so some
player $i$ ends up with $n_i\geq n/(2\ell)=\Omega(n)$ vertices. His
independent set $I_i$ is of size $n_i/\tO(k)=n/\tO(k)$, and $I_i$ is
also independent in $G_0$.  We can therefore use a single color on
$I_i$ and recurse on $G_0\setminus I_i$.

The formal proof that the game leads to a $\tO(k)$ coloring, including
the special monochromatic player, is described in \cite{Blum94}.

We now introduce an SDP player $s$ that for some desired $\Delta$
claims progress if there is any vertex $v$ of degree below $\Delta$,
and moves $v$ to his set $V_s$. The SDP player
will end up with an induced subgraph $G_s=G|V_s$ of $G$ where
the average degree is below $2\Delta$. We delete all
vertices with degree at least $4\Delta$. The
resulting vertex set $V_s'$ has $|V_s'|> n_s/2$ and
$G_s'=G|V_s'$ has maximum degree
$\Delta_{\max}<4\Delta$. If $n_s=|V_s|=\Omega(n)$ and $\Delta=o(n^{0.6546})$,
he can apply Lemma \ref{lem:chlamtac}
to $G_s'$ and get an independent set of size 
$\Omega(n_s'/\Delta_{\max}^{0.3166})=\Omega(n/\Delta^{0.3166})$. The
SDP player thus follows the rule for progress towards $O(\Delta^{0.3166})$
coloring.

Since it for any graph suffices that some player can make progress, a player
like our combinatorial algorithm in Theorem\tref{thm:main} can
wait for the SDP player to remove all the vertices of degree at most
$\Delta$. We therefore only apply Theorem\tref{thm:main} to
graphs with minimum degree $\Delta_{\min}\geq\Delta$.
With Theorem \ref{thm:main}, we then make progress towards
$\tO\left((n/\Delta_{min})^{4/7}\right)=\tO\left((n/\Delta)^{4/7}\right)$ coloring in polynomial time.

Balancing $(n/\Delta)^{4/7}=\Delta^{0.3166}$, we get $\Delta= n^{0.6435}=o(n^{0.6546})$, and conclude
\begin{theorem}\label{thm:sdp}
For any 3-colorable graph $G$,
% with maximum degree $d < n^{0.6546}$,
there is a polynomial time algorithm to make progress towards an $O(n^{0.2038})$-coloring.
\end{theorem}
Hence we get
\begin{corollary}
A 3-colorable graph on $n$ vertices can be colored
with $O(n^{0.2038})$ colors in polynomial time.
\end{corollary}

\bibliographystyle{abbrv} \bibliography{paper}

\begin{thebibliography}{10}

\bibitem{ACC06}
S.~Arora, E.~Chlamtac, and M.~Charikar.
\newblock New approximation guarantee for chromatic number.
\newblock In {\em Proc. 38th STOC}, pages 215--224, 2006.

\bibitem{ARV09}
S.~Arora, S.~Rao, and U.~Vazirani.
\newblock Expanders, geometric embeddings and graph partitioning.
\newblock {\em J. ACM}, 56(2):1--37, 2009.
\newblock Announced at STOC'04.

\bibitem{BR90}
B.~Berger and J.~Rompel.
\newblock A better performance guarantee for approximate graph coloring.
\newblock {\em Algorithmica}, 5(3):459--466, 1990.

\bibitem{Blum94}
A.~Blum.
\newblock New approximation algorithms for graph coloring.
\newblock {\em J. ACM}, 41(3):470--516, 1994.
\newblock Announced at STOC'89 and FOCS'90.

\bibitem{BK97}
A.~Blum and D.~Karger.
\newblock An {${\tilde O}(n^{3/14})$}-coloring algorithm for 3-colorable
  graphs.
\newblock {\em Inf. Process. Lett.}, 61(1):49--53, 1997.

\bibitem{Chl07}
E.~Chlamtac.
\newblock Approximation algorithms using hierarchies of semidefinite
  programming relaxations.
\newblock In {\em Proc. 48th FOCS}, pages 691--701, 2007.

\bibitem{DMR09}
I.~Dinur, E.~Mossel, and O.~Regev.
\newblock Conditional hardness for approximate coloring.
\newblock {\em SIAM J. Comput.}, 39(3):843--873, 2009.
\newblock Announced at STOC'06.

\bibitem{FLS02}
U.~Feige, M.~Langberg, and G.~Schechtman.
\newblock Graphs with tiny vector chromatic numbers and huge chromatic numbers.
\newblock In {\em Proc. 43rd FOCS}, pages 283--292, 2002.

\bibitem{GJS76}
M.~Garey, D.~Johnson, and L.~Stockmeyer.
\newblock Some simplified np-complete graph problems.
\newblock {\em Theor. Comput. Sci.}, 1(3):237--267, 1976.
\newblock Announced at STOC'74.

\bibitem{VGK04}
V.~Guruswami and S.~Khanna.
\newblock On the hardness of 4-coloring a 3-colorable graph.
\newblock {\em SIAM Journal on Discrete Mathematics}, 18(1):30--40, 2004.

\bibitem{KMS98}
D.~Karger, R.~Motwani, and M.~Sudan.
\newblock Approximate graph coloring by semidefinite programming.
\newblock {\em J. ACM}, 45(2):246--265, 1998.
\newblock Announced at FOCS'94.

\bibitem{Karp75}
R.~M. Karp.
\newblock On the computational complexity of combinatorial problems.
\newblock {\em Networks}, 5:45--68, 1975.

\bibitem{KLS00}
S.~Khanna, N.~Linial, and S.~Safra.
\newblock On the hardness of approximating the chromatic number.
\newblock {\em Combinatorica}, 20(3):393--415, 2000.

\bibitem{Sze94}
M.~Szegedy.
\newblock A note on the $\theta$ number of {}{Lov\'asz} and the generalized
  {Delsarte} bound.
\newblock In {\em Proc. 35 FOCS}, pages 36--39, 1994.

\bibitem{Wig83}
A.~Wigderson.
\newblock Improving the performance guarantee for approximate graph coloring.
\newblock {\em J. ACM}, 30(4):729--735, 1983.
\newblock Announced at STOC'82.

\end{thebibliography}

\end{document}